\documentclass[12pt]{article}

\usepackage{amsmath, amssymb, amsthm, physics}
\usepackage{mathrsfs}
\usepackage{bm}
\usepackage{braket}
\usepackage{varwidth}

\usepackage[utf8]{inputenc}
\usepackage[T1]{fontenc}
\usepackage{lmodern}
\usepackage{microtype}
\usepackage{setspace}
\usepackage{hyperref}
\usepackage{epigraph}
\usepackage{algorithm}
\usepackage{algpseudocode}
\usepackage{graphicx}
\usepackage{subcaption}

\usepackage{tikz}
\usetikzlibrary{braids}

\newtheorem{theorem}{Theorem}[section]
\newtheorem{lemma}[theorem]{Lemma}
\newtheorem{corollary}[theorem]{Corollary}

\newcommand{\keywords}[1]{\par\noindent\textbf{Keywords:} #1\par}
\newcommand{\pacs}[1]{\par\noindent\textbf{PACS:} #1\par}

\newcommand{\E}{\mathbb{E}}

\usepackage[
  backend=biber,
  style=phys,
  sorting=none,
  biblabel=brackets
]{biblatex}
\begin{document}

\title{Burau representation, Squier's form, \\ and non-Abelian anyons}
\author{Alexander Kolpakov\footnote{corresponding author; e-mail address: kolpakov.alexander@gmail.com; postal address: University of Austin, 522~Congress Street, Austin TX, 78701 USA}}
\date{November 17, 2026}

\maketitle

\begin{abstract}
We construct a phase tunable family of non-Abelian unitary braid actions
without assuming pre-existing anyonic quasiparticles.  The reduced Burau
representation of $B_3$ is made unitary by Squier's Hermitian form on its
two definiteness intervals.  The resulting two dimensional matrices
satisfy the braid relation and do not commute.  A scalar witness, defined
directly from the trace of their group commutator, is strictly positive
inside these intervals and has an exact expression in the specialization
phase.  Each generator can be realized by ordinary single qubit gates,
so the construction reproduces the specified non-Abelian braid
transformations using conventional quantum gates.  Optimal single use
order discrimination supplies a further interpretation of the witness.
\end{abstract}


\section{Introduction}\label{sec:introduction}

Non-Abelian braid transformations can be studied and implemented without
first realizing particles whose physical exchanges perform them.  The
mathematical requirement is a family of unitary matrices satisfying the
braid relations, with at least two matrices that do not commute.  In this
paper we construct such a family from the reduced Burau representation
of the braid group $B_3$, determine exactly where Squier's invariant form
makes it unitary, and express its generators in terms of ordinary single
qubit gates.

\paragraph{Anyons and braid actions.}
Anyons are excitations in two spatial dimensions whose exchanges are
described by braids rather than just permutations.  An exchange history
of $n$ identical excitations, returning to the same unordered set of
positions, defines an element of the braid group $B_n$.  Its action on a
state space $\mathcal H$ is represented by a homomorphism
\[
\mathcal R:B_n\longrightarrow U(\mathcal H).
\]
The elementary braid $\sigma_i$ exchanges positions $i$ and $i+1$;
concatenation of exchanges becomes multiplication of their matrices.
Products act from right to left.

The particle types are called topological charges.  Their possible
combinations are described by fusion rules
\[
a\times b=\sum_c N_{ab}^{c}c,
\]
where the nonnegative integer $N_{ab}^{c}$ counts the ways charges $a,b$
can combine to give charge $c$.  With the particle types and total charge
fixed, the possible intermediate fusion channels form a fusion space on
which braiding acts~\cite{Kitaev2006,Nayak2008}.  For Abelian anyons this
space is one dimensional, so exchanges act by scalar phases.  For
identical particles the braid relation then gives the same phase for
every elementary exchange, $\mathcal R(\sigma_i)=e^{i\alpha}$.  These
actions commute.  For non-Abelian anyons the fusion space can have
dimension greater than $1$, and exchange matrices can fail to commute:
\[
\mathcal R(\sigma_i)\mathcal R(\sigma_j)
\ne\mathcal R(\sigma_j)\mathcal R(\sigma_i).
\]
For example, Ising anyons obey $\sigma\times\sigma=1+\varepsilon$,
where $1$ is the vacuum charge and $\varepsilon$ is a fermion; four
$\sigma$ particles with vacuum total charge have a two dimensional
fusion space~\cite{Kitaev2006,Nayak2008}.

\paragraph{The present construction.}
We reproduce the noncommuting braid action on a two dimensional target
space $V\cong\mathbb C^2$.  The construction starts with the reduced
Burau representation at $t=e^{i\omega}$.  Squier's Hermitian form
provides a simultaneous change of basis making both generators unitary
on two explicit phase intervals~\cite{Squier1984,Salter2021}.  Their
products obey $U_1U_2U_1=U_2U_1U_2$, while $U_1U_2\ne U_2U_1$.
We quantify the difference between the two orders by a scalar obtained
from the trace of their group commutator.  Its closed expression is
strictly positive inside the allowed region, and each generator has an
explicit decomposition into ordinary rotations and a scalar phase
~\cite{Barenco1995}.
Thus no anyonic medium is required to implement these particular braid
transformations.

This concerns the representation carried by the target space.  A
physical anyon model additionally specifies fusion spaces for different
particle numbers and braid actions compatible with fusion.  The present
construction does not supply this full structure or the protection
against local perturbations that motivates topological quantum
computation~\cite{Nayak2008}.  It gives a direct algebraic realization of
the braid transformations themselves.  Related approaches reproduce
non-Abelian braiding in acoustic and photonic systems and digital quantum
processors~\cite{Zhang2023,Liu2024,Kim2024,Xu2024}.

Sections~\ref{sec:setup} and~\ref{sec:squier} give the Burau matrices and
Squier unitarization, including the distinction between definiteness and
specialization at roots of unity.  Sections~\ref{sec:witness}
and~\ref{sec:main} define and evaluate the algebraic witness, with a short
interpretation through single use order discrimination.
Section~\ref{sec:implementation} gives the gate realization, an optional
controlled order circuit, and reproducible checks.  The accompanying
code is based on the public auxiliary repository~\cite{github-burau-switch}.

\section{Braid group and target space}\label{sec:setup}
Let $V\cong\mathbb{C}^2$ be the target space on which the braid generators act.  Its identification with the state space of a single qubit will give the gate realization in Section~\ref{sec:implementation}.

\paragraph{Braid group on three strands.}
Let $B_3$ denote the braid group on $3$ strands. There are several equivalent definitions of this mathematical object. 

\emph{Topological / combinatorial.} $B_3$ is the group of ambient isotopy classes of geometric braids on $3$ strands in the cylinder $D^2\times[0,1]$, with group law given by vertical concatenation (stacking). Cf. \cite{Birman1974, KasselTuraev2008}.

\emph{Abstract group presentation.} $B_3$ is generated by elementary crossings $\sigma_1,\sigma_2$ subject to the Yang--Baxter relation
\begin{equation}\label{eq:B3-presentation}
B_3\;\cong\;\big\langle \sigma_1,\sigma_2\ \big|\ \sigma_1\sigma_2\sigma_1=\sigma_2\sigma_1\sigma_2 \big\rangle.
\end{equation}

\emph{Mapping class group.} $B_3$ is the mapping class group of an oriented disc with three marked interior points, where homeomorphisms fix the boundary pointwise and may permute the marked points.  Isotopies are taken relative to the boundary~\cite{Birman1974}.

\paragraph{Pictures.} The generators and a sample word are depicted in Figure~\ref{fig:braids}: black curves are strands; time flows upward. We place strand~1 on the left, 2 in the middle, and 3 on the right. Thus, the picture of $\sigma_1$, resp.\ $\sigma_2$, is the positive crossing of strands $(1,2)$, resp.\ $(2,3)$.

\begin{figure}
\begin{center}
\begin{tikzpicture}

  \begin{scope}
    \node at (0.5,0.5) {$\sigma_1$};
    \pic[
      braid/.cd,
      number of strands=3,
      width=0.6cm,
      crossing height=3cm,
      every strand/.style={very thick},
      strand 1/.style={very thick},
      strand 2/.style={very thick},
      strand 3/.style={very thick}
    ] {braid={s_1}};
  \end{scope}

  \begin{scope}[xshift=3cm]
    \node at (0.5,0.5) {$\sigma_2$};
    \pic[
      braid/.cd,
      number of strands=3,
      width=0.6cm,
      crossing height=3cm,
      every strand/.style={very thick},
      strand 1/.style={very thick},
      strand 2/.style={very thick},
      strand 3/.style={very thick}
    ] {braid={s_2}};
  \end{scope}

  \begin{scope}[xshift=6cm]
    \node at (0.5,0.5) {$w=\sigma_1\,\sigma_2\,\sigma_1$};
    \pic[
      braid/.cd,
      number of strands=3,
      width=0.6cm,
      crossing height=1.1cm,
      every strand/.style={very thick},
      strand 1/.style={very thick},
      strand 2/.style={very thick},
      strand 3/.style={very thick}
    ] {braid={s_1 s_2 s_1}};
  \end{scope}
\end{tikzpicture}
\caption{Two braid generators of $B_3$, $\sigma_1$ and $\sigma_2$, and the word $w=\sigma_1\sigma_2\sigma_1$.  The defining relation identifies this word exactly with $\sigma_2\sigma_1\sigma_2$ in $B_3$.}\label{fig:braids}
\end{center}
\end{figure}

\paragraph{Reduced Burau representation for $B_3$.}
Let $t$ first denote an indeterminate.  The reduced Burau representation is the homomorphism
\[
\psi:\ B_3\longrightarrow \mathrm{GL}\!\big(2,\mathbb{Z}[t,t^{-1}]\big)
\]
where $\mathbb{Z}[t,t^{-1}]$ is the ring of Laurent polynomials
$\sum_{k=m}^{n}a_k t^k$ with integer coefficients.  Throughout, $w$
denotes a braid word and $\omega$ denotes the real specialization phase.
The representation is
defined on generators by
\begin{equation}\label{eq:burau-B3}
\psi(\sigma_1)\;=\;
\begin{pmatrix}
 -t & 1 \\[2pt] 0 & 1
\end{pmatrix},
\qquad
\psi(\sigma_2)\;=\;
\begin{pmatrix}
 1 & 0 \\[2pt] t & -t
\end{pmatrix},
\end{equation}
and extended by way of homomorphism:
\begin{equation}
    \psi(\sigma_{i_1} \sigma_{i_2}\ldots\sigma_{i_k}) = \psi(\sigma_{i_1})\,\psi(\sigma_{i_2})\,\ldots\,\psi(\sigma_{i_k}).
\end{equation}

For each numerical value $t_0\in\mathbb{C}\setminus\{0\}$, evaluation at $t=t_0$ gives a specialized representation $\psi_{t_0}:B_3\to\mathrm{GL}(2,\mathbb{C})$.  The unitary construction below uses the specialization $t_0=e^{i\omega}$; statements about the Laurent polynomial representation and statements about an individual specialization are therefore kept distinct.

The generators given above are the standard $2\times2$ reduced Burau matrices; see \cite[\S~3]{Birman1974} or \cite[\S~3.1]{KasselTuraev2008}.

\section{Squier’s unitarization}\label{sec:squier}
Let $s$ be an indeterminate with involution $\overline{s}=s^{-1}$ and set $t=s^2$. Squier’s modification of the reduced Burau representation is the homomorphism
\[
\beta:\ B_3\longrightarrow \mathrm{GL}\!\big(2,\mathbb{Z}[s,s^{-1}]\big),
\qquad
\sigma_i\longmapsto \beta_i = \beta(\sigma_i),
\]
defined on generators by
\begin{equation}\label{eq:squier-burau-B3}
\beta_1 = \beta(\sigma_1)\;=\;
\begin{pmatrix}
 -s^2 & s\\[2pt] 0 & 1
\end{pmatrix},
\qquad
\beta_2 = \beta(\sigma_2)\;=\;
\begin{pmatrix}
 1 & 0\\[2pt] s & -s^2
\end{pmatrix}.
\end{equation}
We extend the coefficient involution $\overline{s}=s^{-1}$ to matrices
by the adjoint $M^\ast:=\overline{M}^{\,T}$.

Squier exhibits the Hermitian form
\begin{equation}\label{eq:squier-J-symbolic}
J(s)\;=\;(s+s^{-1})\,I_2\;-\;\begin{pmatrix}0&1\\[2pt]1&0\end{pmatrix}
\;=\;\begin{pmatrix}s+s^{-1}&-1\\[2pt]-1&s+s^{-1}\end{pmatrix},
\end{equation}
and proves the identities
\begin{equation}\label{eq:J-unitary}
\beta_i(s)^\ast\,J(s)\,\beta_i(s)\;=\;J(s)\qquad(i=1,2).
\end{equation}
For real $\omega$, specialize $s=e^{i\omega/2}$, so $t=s^2=e^{i\omega}$,
and write $\beta_i(\omega):=\left.\beta_i(s)\right|_{s=e^{i\omega/2}}$ and
$J(\omega):=\left.J(s)\right|_{s=e^{i\omega/2}}$.  The coefficient involution now becomes
complex conjugation, and $M^\ast$ becomes the ordinary adjoint
$M^\dagger$.  Thus \eqref{eq:J-unitary} states
$\beta_i(\omega)^\dagger J(\omega)\beta_i(\omega)=J(\omega)$.

\paragraph{Relation to the standard reduced Burau.}
Let us set
\[
D(s)=\mathrm{diag}(1,s^{-1}).
\]
Then
\begin{equation}\label{eq:similarity}
\psi(\sigma_i)\;=\;D(s)^{-1}\,\beta(\sigma_i)\,D(s),\qquad
\tilde{J}(s)\;=\;D(s)^\ast\,J(s)\,D(s),
\end{equation}
and consequently $\psi(\sigma_i)^\ast\,\tilde{J}(s)\,\psi(\sigma_i) = \tilde{J}(s)$.

We use Squier's basis because its real symmetric specialized form has
explicit eigenvalues, making the signature intervals and positive square
root transparent.  Equation \eqref{eq:similarity} transfers this
calculation to the standard Burau basis.

\paragraph{Unitary specializations and the control parameter.}
We regard $\omega$ as a phase parameter modulo $2\pi$.  Because the modified Burau matrices use the square root $s=e^{i\omega/2}$, we work on the open representative
\begin{equation}\label{eq:omega-domain}
\omega \in (0,2\pi).
\end{equation}
The endpoints give the same $t=1$, but $s=1$ at $\omega=0$ and $s=-1$
at $\omega=2\pi$.  We omit them to keep an open parameter interval;
neither is a singularity of Squier's form, as shown below.  Specializing at
$s=e^{i\omega/2}$ (hence $t=s^2=e^{i\omega}$), the Squier form becomes
\begin{equation}\label{eq:J-omega}
J(\omega)\;=\;\begin{pmatrix}2\cos(\omega/2)&-1\\[2pt]-1&2\cos(\omega/2)\end{pmatrix}.
\end{equation}
Its eigenvalues are
\begin{equation}\label{eq:J-eigenvalues}
\lambda_-(\omega)=2\cos(\omega/2)-1,
\qquad
\lambda_+(\omega)=2\cos(\omega/2)+1,
\end{equation}
Figure~\ref{fig:witness-gap}(a) plots both eigenvalues.  On
$(0,2\pi/3)$ they are both positive; on $(2\pi/3,4\pi/3)$ one has
$\lambda_-<0<\lambda_+$; and on $(4\pi/3,2\pi)$ they are both negative.
At $\omega=2\pi/3$ one has $\lambda_-=0$, while at $\omega=4\pi/3$ one
has $\lambda_+=0$.  Thus $J$ is positive definite on the first interval,
indefinite on the middle interval, and negative definite on the final
interval.  Both definite intervals yield ordinary unitary matrices after a
sign normalization.  We consequently define the \emph{Squier definiteness},
or Euclidean unitarity, region
\begin{equation}\label{eq:pd-window}
\Omega_+ \;=\; (0,2\pi/3)\ \cup\ (4\pi/3,2\pi),
\end{equation}
and set
\begin{equation}\label{eq:positive-form}
\varepsilon(\omega)=
\begin{cases}
+1,&0<\omega<2\pi/3,\\
-1,&4\pi/3<\omega<2\pi,
\end{cases}
\qquad
H(\omega):=\varepsilon(\omega)J(\omega)\succ0.
\end{equation}
Here $H\succ0$ means $v^\dagger Hv>0$ for every nonzero vector $v$.
Since multiplication of an invariant Hermitian form by $-1$ does not change its invariance, one has
\begin{equation}\label{eq:H-unitary}
\beta_i(\omega)^\dagger H(\omega)\beta_i(\omega)=H(\omega),
\qquad i=1,2.
\end{equation}

On either connected component of $\Omega_+$, choose the positive square root
$R(\omega)=H(\omega)^{1/2}$, so that $H(\omega)=R(\omega)^\dagger R(\omega)$.
Write $\beta_\omega(w)=\beta(w)|_{s=e^{i\omega/2}}$.  For every braid word
$w\in B_3$, define
\begin{equation}\label{eq:eucl-unitary}
\mathcal U_\omega(w)
:=R(\omega)\,\beta_\omega(w)\,R(\omega)^{-1}\in U(2).
\end{equation}
Unitarity follows from \eqref{eq:H-unitary}.

\subsection{Definiteness, pseudounitarity, and roots of unity}\label{sec:positivity-window}

On the complementary interval $(2\pi/3,4\pi/3)$, the invariant form
$J(\omega)$ has mixed signature $(1,1)$.  The specialized matrices
preserve this indefinite form (pseudounitarity), but the conjugation
\eqref{eq:eucl-unitary} requires a positive form.

\paragraph{Relation to roots of unity truncations in braided Majorana qubit models.}

In braided Majorana qubit models, certain roots of unity give
truncations: some higher excitation states are
absent~\cite{Toppan2022Majorana,Toppan2024Volichenko}.  Here
they enter only through the specialization $t=e^{i\omega}$.  The additional
condition $\omega\in\Omega_+$ is specific to realizing that specialization by
ordinary Hilbert space unitaries through \eqref{eq:eucl-unitary}.

To make both generators unitary, one needs the same positive invariant
form for both.  Individual finite order does not ensure that such a
common form exists.  This can be seen from the algebraic commutator
before unitarization:
\begin{equation}\label{eq:algebraic-commutator}
K_\beta:=\beta_2^{-1}\beta_1^{-1}\beta_2\beta_1
=\begin{pmatrix}
1-s^2&-s^{-1}\\[2pt]
s^{-1}&-s^{-2}
\end{pmatrix}.
\end{equation}
Let $K_\beta(\omega)$ denote this matrix at $s=e^{i\omega/2}$.  It satisfies
\begin{equation}\label{eq:algebraic-commutator-invariants}
\det K_\beta=1,
\qquad
\operatorname{tr}K_\beta=1-2\cos\omega.
\end{equation}
On $(2\pi/3,4\pi/3)$ this trace exceeds $2$, whereas a $2\times2$ unitary
with determinant $1$ has trace in $[-2,2]$.  Thus no common similarity can make both
generators unitary there, even if a different positive invariant form is
sought for the same full pair.

For example, the seventh roots $t_k=e^{2\pi i k/7}$, $k=1,2,3$, give
\[
\begin{array}{c|r|r}
k & \det J(2\pi k/7) & \operatorname{tr}K_\beta(2\pi k/7)\\ \hline
1 & 2.246979604 & -0.246979604\\
2 & 0.554958132 & 1.445041868\\
3 & -0.801937736 & 2.801937736
\end{array}
\]
All three have elementary generators whose $14$th powers are the
identity, but the third pair admits no simultaneous similarity into
$U(2)$.  The first two pairs correspond to phases in $\Omega_+$, but their commutator traces
differ, so the pairs are not simultaneously similar.  Individual finite
order therefore does not determine the full representation.  These
statements concern the full
specialized pair; they do not exclude other quantum constructions,
multiparticle models, or quotient representations at these parameters.

\paragraph{Boundary points.}
At $\omega=2\pi/3$ and $4\pi/3$, respectively, $\lambda_-$ and
$\lambda_+$ vanish.  The limiting form is singular, so an invertible
positive square root factor for \eqref{eq:eucl-unitary} is unavailable.
The specialized Burau matrices remain invertible.  Their group
commutator \eqref{eq:algebraic-commutator} has both eigenvalues equal to $1$
at either inner boundary, but is not the identity.  It is a nontrivial
unipotent matrix and cannot be similar to a unitary, ruling out any common
positive invariant form for the full boundary pair as well.  At $\omega=0$ and
$2\pi$, the eigenvalues of $J$ are $\{1,3\}$ and $\{-3,-1\}$,
respectively, so the form is nonsingular.  These endpoints both represent
$t=1$.  With $Q=\operatorname{diag}(1,-1)$,
\[
\beta_i(2\pi)=Q\beta_i(0)Q,
\qquad
-J(2\pi)=QJ(0)Q,
\]
so their limiting unitary realizations differ only by a change of basis.

\begin{lemma}[Squier definiteness intervals]\label{lem:squier}
On each connected component $I\subset\Omega_+$, the sign $\varepsilon(\omega)$ is constant and the positive square root $R(\omega)=H(\omega)^{1/2}$ is real-analytic.  Consequently
\[
\omega\longmapsto\mathcal U_\omega(w)
\]
is a real-analytic function with values in $U(2)$ for every fixed $w\in B_3$.
\end{lemma}

\begin{proof}
On either component of $\Omega_+$, $H(\omega)$ is a real-analytic positive definite matrix whose spectrum stays away from zero on compact subintervals.  Its unique positive square root is therefore real-analytic~\cite[Chapter~II, \S 6.1--6.4]{Kato1995}.  Conjugation preserves analyticity, and \eqref{eq:H-unitary} implies that the resulting matrices are unitary.
\end{proof}

\section{An algebraic witness for the braid orders}\label{sec:witness}

For $\omega\in\Omega_+$, let
\begin{equation}\label{eq:unitary-generators}
U_i(\omega)=\mathcal U_\omega(\sigma_i),\qquad i=1,2,
\end{equation}
and denote the two ordered products by
\begin{equation}\label{eq:two-orders}
V_{12}(\omega)=U_1(\omega)U_2(\omega),
\qquad
V_{21}(\omega)=U_2(\omega)U_1(\omega).
\end{equation}
Their relative transformation is the group commutator
\begin{equation}\label{eq:relative-commutator}
K(\omega):=V_{12}(\omega)^\dagger V_{21}(\omega)
=U_2^{-1}U_1^{-1}U_2U_1.
\end{equation}
It is unitary and has determinant $1$, so $K\in SU(2)$.  Its eigenvalues
therefore have the form $e^{i\gamma},e^{-i\gamma}$ for
$\gamma\in[0,\pi]$, and $\operatorname{tr}K=2\cos\gamma$ is real.

We define the algebraic order witness by
\begin{equation}\label{eq:NA-witness}
\mathcal W_{\rm NA}(\omega)
:=\frac12\sqrt{1-\frac{|\operatorname{tr}K(\omega)|^2}{4}}.
\end{equation}
This is a real number in $[0,1/2]$ and is unchanged by a common unitary
change of basis of $U_1,U_2$.  It is zero precisely when $K=I$ or $K=-I$:
a unitary matrix is diagonalizable, so trace $2$ or $-2$ forces the
corresponding scalar matrix.  Equivalently, the witness vanishes when
\[
U_1U_2=e^{i\varphi}U_2U_1
\quad\text{for some }\varphi\in\mathbb R.
\]
Indeed, determinants require $e^{2i\varphi}=1$.  Thus a positive witness
shows that the two orders differ even up to a global phase, and in
particular that the generators do not commute.  The qualification about
global phase matters: anticommuting matrices would also give a zero
witness.  The theorem below rules out both possibilities for the present
Burau family inside $\Omega_+$.

\section{The exact non-Abelian braid action}\label{sec:main}

\begin{theorem}[Exact braid order witness]\label{thm:non-abelian}
For every $\omega\in\Omega_+$, the matrices $U_1(\omega)$ and
$U_2(\omega)$ in \eqref{eq:unitary-generators} are unitary, satisfy
\begin{equation}\label{eq:unitary-braid-relation}
U_1U_2U_1=U_2U_1U_2,
\end{equation}
and do not commute.  They define a non-Abelian unitary representation of
$B_3$ on $V$.  Their algebraic order witness is
\begin{equation}\label{eq:NA-witness-closed}
\mathcal W_{\rm NA}(\omega)
=\frac12\sqrt{1-\left(\frac12-\cos\omega\right)^2}>0
\qquad(\omega\in\Omega_+).
\end{equation}
\end{theorem}

\begin{proof}
Unitarity follows from \eqref{eq:H-unitary}.  The same similarity
transformation is applied to both Burau generators, so it preserves
their braid relation.  Direct multiplication before unitarization gives
\begin{equation}\label{eq:generator-commutator-difference}
\beta_1\beta_2-\beta_2\beta_1
=\begin{pmatrix}
s^2&-s(1+s^2)\\
s(1+s^2)&-s^2
\end{pmatrix}\ne0,
\end{equation}
since $s\ne0$.  Thus $U_1$ and $U_2$ do not commute either.

For the witness, the commutator in \eqref{eq:relative-commutator}
satisfies $K=R K_\beta R^{-1}$, where $K_\beta$ was calculated in
\eqref{eq:algebraic-commutator}.  Similarity preserves trace and
determinant, giving
\begin{equation}\label{eq:K-invariants}
\det K=1,\qquad
\operatorname{tr}K=1-s^2-s^{-2}=1-2\cos\omega.
\end{equation}
Substitution in \eqref{eq:NA-witness} proves the closed expression.
Inside $\Omega_+$, one has $-1/2<\cos\omega\leq1$ and hence
$-1<1/2-\cos\omega<1$, which proves strict positivity.
\end{proof}

The witness approaches $0$ at the singular inner boundaries
$\omega=2\pi/3$ and $4\pi/3$, and reaches its maximum $1/2$ at
$\omega=\pi/3$ and $5\pi/3$.  Varying $\omega$ therefore changes the
distinction between the two braid orders continuously.  Every phase
inside $\Omega_+$ gives a non-Abelian action; a scalar Abelian action is
not attained there.

\subsection{Interpretation by single use order discrimination}\label{sec:discrimination}

The witness also has a simple operational interpretation, independent of
its algebraic definition.  A unitary matrix $U$ induces the quantum
channel $\rho\mapsto U\rho U^\dagger$, where $\rho$ is a positive
semidefinite matrix with trace $1$ describing the target state.
Suppose an unknown device applies $V_{12}$ or $V_{21}$ with equal
probability.  We may choose a normalized input vector $\ket\xi\in V$
and measure the target output once to identify which was applied.
The bra $\bra\xi=\ket\xi^\dagger$ denotes the adjoint of the input ket.
A binary measurement consists of positive semidefinite operators
$E_{12},E_{21}$ with $E_{12}+E_{21}=I$; outcome $ij$ means guessing
$V_{ij}$.  For a normalized output vector $\ket\phi$, that outcome has
probability $\langle\phi|E_{ij}|\phi\rangle$.
The optimal success probability is
\[
p^*(V_{12},V_{21}):=
\max_{\substack{\|\xi\|=1\\E_{12},E_{21}\succeq0\\E_{12}+E_{21}=I}}
\frac12\left(
\langle\xi|V_{12}^\dagger E_{12}V_{12}|\xi\rangle+
\langle\xi|V_{21}^\dagger E_{21}V_{21}|\xi\rangle\right).
\]
Here ``single use'' means one application of the unknown channel.
Random guessing succeeds with probability $1/2$, so $p^*\geq1/2$;
$p^*=1$ means perfect discrimination.

\begin{corollary}\label{cor:discrimination}
For the two braid orders,
\begin{equation}\label{eq:helstrom}
p^*(V_{12},V_{21})=\frac12+\mathcal W_{\rm NA}(\omega).
\end{equation}
In particular, the channels are perfectly distinguishable in a single
use at $\omega=\pi/3$ and $5\pi/3$.
\end{corollary}

\begin{proof}
For a fixed input, the output vectors have overlap
$\langle\xi|K|\xi\rangle$.  Helstrom's formula for two pure states with
equal priors gives the optimized measurement success probability
$\tfrac12(1+\sqrt{1-|\langle\xi|K|\xi\rangle|^2})$
~\cite{Helstrom1976,Acin2001}.  Choose orthonormal eigenvectors
$\ket{v_+},\ket{v_-}$ of $K$ with eigenvalues $e^{\pm i\gamma}$.
For $\ket\xi=c_+\ket{v_+}+c_-\ket{v_-}$ and
$q=|c_+|^2$, one has
\[
\langle\xi|K|\xi\rangle
=q e^{i\gamma}+(1-q)e^{-i\gamma}
=\cos\gamma+i(2q-1)\sin\gamma.
\]
Every $q\in[0,1]$ occurs for a normalized input.  The set of these
overlaps is called the numerical range of $K$; the display shows that it
is the line segment joining the two eigenvalues
~\cite[Property~1.2.9]{HornJohnson1991}.  Its minimum modulus is
$|\cos\gamma|=|\operatorname{tr}K|/2$, attained at $q=1/2$.
Substituting this minimum into Helstrom's formula proves
\eqref{eq:helstrom}.
\end{proof}

\section{Realization with ordinary quantum gates}\label{sec:implementation}

\subsection{Single qubit rotations}

To implement the generators on a single qubit, we express the matrices
$U_i=H^{1/2}\beta_i H^{-1/2}$ as products of rotations and a scalar
phase.  Write
\[
\sigma_y=\begin{pmatrix}0&-i\\i&0\end{pmatrix},\qquad
\sigma_z=\begin{pmatrix}1&0\\0&-1\end{pmatrix},\qquad
R_\mu(x)=e^{-ix\sigma_\mu/2}\quad(\mu=y,z).
\]
Each generator has an Euler decomposition
\begin{equation}\label{eq:euler-decomposition}
U_i=e^{i\eta_i}R_z(a_i)R_y(b_i)R_z(c_i),
\end{equation}
with real angles $\eta_i,a_i,b_i,c_i$
~\cite[Lemma~4.1]{Barenco1995}.  Thus a braid word is implemented by
concatenating these rotation sequences; a negative power uses the
inverse sequence.  Both elementary generators carry the same
determinant $-e^{i\omega}$, so their scalar phases may be chosen equal.
Keeping these phases gives exactly the displayed representation.  The
rotations alone give the same transformations up to a global phase and
preserve both the braid relation and the witness.

For completeness, the angles can be obtained directly from any of the
matrices $U=U_i$.  Choose $\eta=\tfrac12\arg\det U$ and write
\[
e^{-i\eta}U=\begin{pmatrix}u&-\overline v\\v&\overline u\end{pmatrix},
\qquad |u|^2+|v|^2=1.
\]
When $u,v\ne0$, a choice is
\begin{equation}\label{eq:euler-angles}
\begin{aligned}
b&=2\operatorname{atan2}(|v|,|u|),\\
a&=\arg v-\arg u,\qquad
c=-\arg v-\arg u.
\end{aligned}
\end{equation}
Here $\operatorname{atan2}(y,x)$ is the angle in $[0,\pi/2]$ whose
sine and cosine are proportional to the nonnegative numbers $y,x$.
If $v=0$, take $b=c=0$ and $a=-2\arg u$; if $u=0$, take
$b=\pi$, $c=0$, and $a=2\arg v$.  Multiplication of the three rotation
matrices verifies \eqref{eq:euler-decomposition}.  The angles need not be
continuous across choices of argument, although the matrices are
real-analytic on each component of $\Omega_+$.

\subsection{An optional controlled order circuit}

The two orders can also be placed in a quantum switch
~\cite{Chiribella2013,Goswami2018}.  For an independent control qubit
$C$ with orthonormal basis $\{\ket0,\ket1\}$, let $\ketbra{j}{j}_C$
denote the orthogonal projector onto $\ket j$.  Define the controlled
order matrix
\begin{equation}\label{eq:controlled-braid-order}
\mathsf S_\omega
=\ketbra{0}{0}_C\otimes U_1U_2
+\ketbra{1}{1}_C\otimes U_2U_1.
\end{equation}
Preparing the control in
$\ket+=(\ket0+\ket1)/\sqrt2$ gives
\[
\mathsf S_\omega(\ket+\otimes\ket\xi)
=\frac{\ket0\otimes U_1U_2\ket\xi+
\ket1\otimes U_2U_1\ket\xi}{\sqrt2}.
\]
This coherently controls the order of the two transformations.  The
switch is usually considered as a higher order operation taking two
channels as inputs.  With $U_1,U_2$ fixed, \eqref{eq:controlled-braid-order}
is a specified unitary on $C\otimes V$ and can also be synthesized from
ordinary single qubit gates and controlled NOT gates
~\cite{Barenco1995}.  A controlled NOT gate acts by
$\ket a\ket b\mapsto\ket a\ket{a+b\bmod2}$ for $a,b\in\{0,1\}$.
The braid representation and its witness do not
depend on choosing this implementation.

The causal properties of a switch with variable input operations are
studied using process matrices, which describe how local quantum
operations are connected without presupposing a global causal order
~\cite{Oreshkov2012}.  Causal nonseparability means that a process cannot
be written as a probabilistic mixture of processes with the two
opposite definite orders~\cite{Araujo2015,Branciard2016}.  The witness
\eqref{eq:NA-witness} concerns the two specified braid products; it does
not test that property of a process.

\subsection{Numerical verification}

The accompanying reproducer checks the exact form invariance, braid
relation, noncommutation, and commutator trace and determinant
symbolically.  On a phase grid in both components of $\Omega_+$ it
checks ordinary unitarity after the common similarity transformation,
evaluates the witness directly from the resulting matrices, and
compares it with \eqref{eq:NA-witness-closed}.  It also reconstructs the
generators from \eqref{eq:euler-decomposition} and checks the
controlled order matrix \eqref{eq:controlled-braid-order}.  These
calculations are reproducible from the accompanying code; the public
auxiliary repository provides the underlying implementation
~\cite{github-burau-switch}.

\begin{figure}[htbp]
\centering
\includegraphics[width=0.99\textwidth]{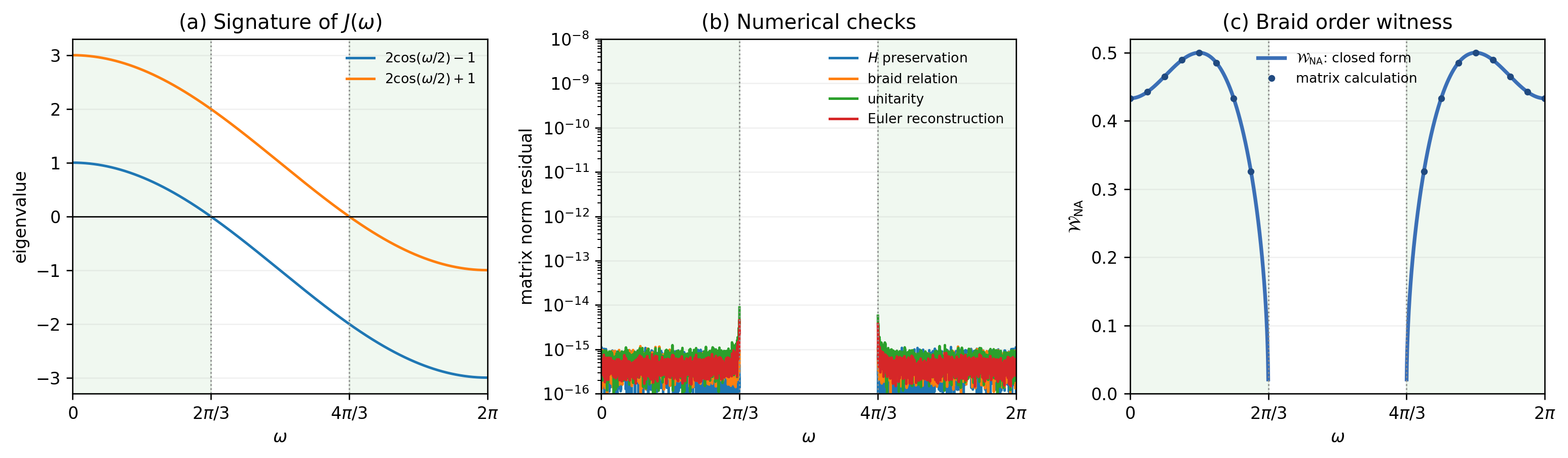}
\caption{(a) Eigenvalues of $J(\omega)$; the two definite intervals give
$\Omega_+$ after the sign normalization $H=\varepsilon J$.
(b) Numerical residuals for form invariance, the braid relation,
ordinary unitarity, and Euler reconstruction.
(c) The algebraic order witness: its closed expression and values
calculated from the unitarized matrices.}\label{fig:witness-gap}
\end{figure}

The numerical checks agree with the exact identities at floating-point
precision.  The Squier similarity and Euler decomposition give a
construction of the specified non-Abelian $B_3$ action with ordinary
qubit gates.

\begingroup
\emergencystretch=1em
\printbibliography
\endgroup

\end{document}